\documentclass[letterpaper, 10 pt, journal, twoside]{IEEEtran}
\IEEEoverridecommandlockouts

\usepackage{cite}
\usepackage{amsmath,amssymb,amsfonts}
\usepackage{bm}
\usepackage{graphicx}
\usepackage{textcomp}
\usepackage{xcolor,color}
\usepackage{amsthm}
\usepackage[top=0.833in, left=0.667in, right=0.667in, bottom=0.597in, includefoot]{geometry}
\usepackage{hyperref}
\usepackage{tikz}
\usepackage{tikz-qtree}
\usetikzlibrary{positioning,automata,backgrounds}
\usepackage{algorithm}
\usepackage[noend]{algpseudocode}
    
\makeatletter
\let\MYcaption\@makecaption
\makeatother 

\usepackage[font=footnotesize]{subcaption}

\makeatletter
\let\@makecaption\MYcaption
\makeatother

\newcommand{\until}{\mathcal{U}}
\newcommand{\always}{\square}
\newcommand{\eventually}{\lozenge}

\newcommand{\x}{\bm{x}}
\renewcommand{\u}{\bm{u}}
\newcommand{\y}{\bm{y}}
\newcommand{\A}{\bm{A}}
\newcommand{\B}{\bm{B}}
\newcommand{\C}{\bm{C}}
\newcommand{\D}{\bm{D}}
\newcommand{\Q}{\bm{Q}}
\newcommand{\R}{\bm{R}}

\newtheorem{definition}{Definition}

\newtheorem{theorem}{Theorem}
\newtheorem{proposition}{Proposition}
\newtheorem{remark}{Remark}
\newtheorem{problem}{Problem}

\newtheorem{example}{Example}

\newcommand{\hl}[1]{#1}

\begin{document}

\title{Mixed-Integer Programming for Signal Temporal Logic with Fewer Binary Variables}

\author{Vince Kurtz and Hai Lin
\thanks{The authors are with the Departments of Electrical Engineering, University of Notre Dame, Notre Dame, IN, 46556 USA. \texttt{\{vkurtz,hlin1\}@nd.edu}}
\thanks{This work was supported by NSF Grants CNS-1830335, IIS-2007949.}
}

\maketitle
\thispagestyle{empty}

\begin{abstract}
    Signal Temporal Logic (STL) provides a convenient way of encoding complex control objectives for robotic and cyber-physical systems. The state-of-the-art in trajectory synthesis for STL is based on Mixed-Integer Convex Programming (MICP). The MICP approach is sound and complete, but has limited scalability due to exponential complexity in the number of binary variables. In this letter, we propose a more efficient MICP encoding for STL. Our new encoding is based on the insight that disjunction can be encoded using a logarithmic number of binary variables and conjunction can be encoded without binary variables. We demonstrate in simulation examples that our proposed approach significantly outperforms the state-of-the-art for long and complex specifications. \hl{Open-source software is available: \texttt{\url{https://stlpy.readthedocs.io}}. }
\end{abstract}

\section{Introduction}\label{sec:intro}

\IEEEPARstart{S}{ignal} Temporal Logic (STL) is a powerful means of expressing complex control objectives. STL combines boolean operators (``and'', ``or'', ``not'') with temporal operators (``always'', ``eventually'', ``until'') and is defined over continuous-valued signals, making it an appealing choice for dynamical systems ranging from mobile robots \cite{lindemann2019coupled,gilpin2021smooth,mehdipour2019arithmetic} and quadrotors \cite{pant2018fly} to high-DoF manipulators \cite{kurtz2020trajectory} and traffic networks \cite{sadraddini2016model,coogan2015traffic}.

In this letter, we address the trajectory synthesis problem for discrete-time linear systems subject to STL specifications over convex predicates. That is, given a system and a specification, find a satisfying trajectory \hl{that minimizes a given cost.}

The state-of-the-art solution uses Mixed-Integer Convex Programming (MICP) \cite{belta2019formal}. This approach, first proposed in \cite{raman2014model} and later refined in \cite{sadraddini2015robust,sadraddini2018formal}, is sound and complete: any solution returned by the algorithm satisfies the specification, and the algorithm finds a solution whenever one exists. Furthermore, MICP finds globally optimal solutions, enabling maximally robust or minimum energy trajectories. 

The major limitation is scalability. MICP's worst-case complexity is exponential in the number of binary variables, and the state-of-the-art in STL synthesis introduces a binary variable for each predicate at each timestep. This limits MICP methods to simple specifications and short time horizons \cite{belta2019formal}. 

This exponential complexity should not be too surprising. STL synthesis is NP-hard, so exponential worst-case complexity is inevitable for any sound and complete algorithm (assuming $P \neq NP$). But despite this NP-hardness, MICP performs well in practice for moderately sized problems. With this in mind, much research in STL synthesis has gone toward reducing MICP problem sizes. For example, rewriting formulas in Positive Normal Form (PNF) significantly reduces the number of binary variables and constraints \cite{belta2019formal,sadraddini2018formal}. Waypoint-based abstractions can further limit the problem size \cite{sun2022multi}.

\begin{figure}
    \centering
    \begin{tikzpicture}[framed]
    \hl{
        \Tree [.{$\always_{[1,T]} a \land \eventually_{[1,T]} b$ }
                [.{$\always_{[1,T]} a$} 
                 [.{\hl{$a_1$}} ] 
                 [.{\hl{$a_2$}} ] 
                 [.{$\dots$} ]
                 [.{\hl{$a_{T}$}} ] 
                ]
                [.{$\eventually_{[1,T]} b$}
                 [.{\hl{$b_1$}} ] 
                 [.{\hl{$b_2$}} ] 
                 [.{$\dots$} ]
                 [.{\hl{$b_{T}$}} ] 
                ]
              ]}
    \end{tikzpicture}
    
    \caption{Tree representation of the STL formula \hl{$\always_{[1,T]}a \land \eventually_{[1,T]} b$}. For this example, our approach reduces the number of binary variables from $2T$ to $\lceil \log_2(T+1) \rceil$, reducing computational complexity from $O(2^T)$ to $O(T)$.}
    \label{fig:simple_tree}
\end{figure}
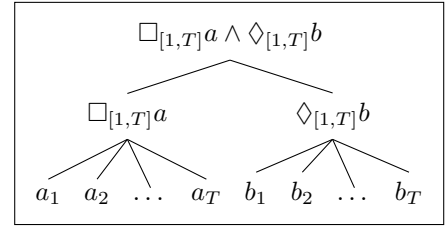

Another possibility is to formulate MICPs with tighter convex relaxations. While this does not improve the worst-case complexity, a tighter convex relaxation increases branch-and-bound efficiency and often improves performance in practice \cite{conforti2014integer}. This approach has been pursued for Metric Temporal Logic in \cite{kurtz2021more} and piecewise-affine reachability in \cite{marcucci2021shortest}. 

Much STL research in recent years has focused on avoiding MICP entirely, and instead designing efficient heuristic methods which are sound but not complete. Such methods include gradient-based optimization \cite{pant2017smooth,mehdipour2019arithmetic,gilpin2021smooth,leung2020back}, differential dynamic programming \cite{kurtz2020trajectory}, control barrier functions \cite{lindemann2018control} and neural-network-based methods \cite{liu2021recurrent}. These approaches tend to improve scalability, especially with respect to the specification time horizon, but offer limited formal guarantees and often struggle to handle more complex specifications. 

In this letter, we return to STL synthesis via MICP and propose a new MICP encoding that uses fewer binary variables. Our proposed approach draws on a logarithmic encoding of SOS1 constraints \cite{vielma2011modeling} to reduce the number of binary variables. \hl{Similar to \cite{belta2019formal,sadraddini2018formal}, we consider formulas in PNF,} but rather than introducing a binary variable for each predicate at each timestep, our proposed approach introduces binary variables only for disjunctive subformulas. Furthermore, each disjunctive subformula is encoded using a logarithmic number of binary variables. This approach leads to particularly notable reductions in complexity for long-horizon specifications, \hl{though the standard MICP encoding \cite{belta2019formal} tends to outperform our proposed approach for short horizons}.

For example, consider the specification shown in Figure~\ref{fig:simple_tree}, which reads ``always $a$ and eventually $b$''. The standard MICP encoding requires a linear number of binary variables with respect to time horizon $T$, resulting in exponential computational complexity. Our proposed encoding uses a logarithmic number of binary variables, resulting in linear complexity. 

Our primary contributions are as follows:
\begin{itemize}
    \item We propose a new MICP encoding for STL specifications that uses fewer binary variables for most specifications,
    \item Our proposed approach is sound and complete, and guaranteed to find a globally optimal solution,
    \item We propose a tree data structure for STL formulas, and show how this data structure can be exploited to further reduce the number of binary variables,
    \item We demonstrate the scalability of our approach in various simulated scenarios involving a mobile robot. \hl{Code for reproducing these results is available online \cite{stlpy}.}
\end{itemize}

The remainder of this letter is organized as follows: background and a formal problem statement are presented in Section~\ref{sec:background}, our main results are presented in Section~\ref{sec:main}, and simulation experiments are reported in Section~\ref{sec:experimenal}. 

\section{Background}\label{sec:background}

\subsection{System Definitions}

In this letter, we consider discrete-time linear systems:
\begin{equation}\label{eq:sys}
\begin{gathered}
    \x_{t+1} = \A \x_t + \B \u_t, \\
    \y_t = \C \x_t + \D \u_t,
\end{gathered}
\end{equation}
where $\x_t \in X \subseteq \mathbb{R}^n$ is the system state at timestep $t$, $\u_t \in {U} \subseteq \mathbb{R}^m$ is the control input, and $\y_t \in {Y} \subseteq \mathbb{R}^p$ is the output. We assume that ${X}$, ${U}$ and ${Y}$ are convex sets, and $\A,\B,\C,\D$ are matrices of appropriate dimension.

Given an initial state $\x_0$ and a control tape $\u = \u_0,\u_1,\dots$, the output signal $\y = \y_0,\y_1,\dots$ is generated by applying (\ref{eq:sys}). We denote the output signal starting at timestep $t$ by
\begin{equation*}
    (\y,t) = \y_t,\y_{t+1},\y_{t+2},\dots.
\end{equation*}

\subsection{Signal Temporal Logic}\label{subsec:stl}

We consider STL formulas over convex predicates in PNF: further details and discussion of STL formulas can be found in \cite{belta2019formal}. Such formulas follow the syntax
\begin{equation}\label{eq:syntax}
   \varphi := \pi \mid \bigvee_i \varphi_i \mid \bigwedge_i \varphi_i \mid \always_{[t_1,t_2]} \varphi \mid \eventually_{[t_1,t_2]} \varphi \mid \varphi_1 \until_{[t_1,t_2]} \varphi_2,
\end{equation}
where predicates $\pi$ are defined by convex functions $g^\pi : {Y} \to \mathbb{R}$ \hl{such that $\{\y \mid g^\pi(\y) \leq 0\}$ are convex sets}. In addition to standard boolean operations ``and'' ($\land$) and ``or'' ($\lor$), STL includes temporal operators ``always'' ($\always$), ``eventually'' ($\eventually$) and ``until'' ($\until$). In this letter we consider only bounded-time specifications, i.e., those with $t_2$ finite. 

\begin{remark}
    This syntax is slightly more general than that with linear predicates only, as considered in \cite{raman2014model,sadraddini2015robust,sadraddini2018formal,sun2022multi}, since such convex predicates include linear ones $g^\pi(\y_t) = a^T\y - b$. Note that negation ($\lnot$) can be applied to linear predicates, but not to more general convex predicates such as ellipses. 
\end{remark}

The semantics, or meaning, of STL formulas are defined as follows, where we denote the fact that signal $\y$ satisfies specification $\varphi$ as $\y \vDash \varphi$. Here we present the quantitative semantics of STL, which defines a scalar ``robustness value'' $\rho^{\varphi}(\y)$ which is positive only if $\y \vDash \varphi$:
\begin{itemize}
    \item $\y \vDash \varphi$ if and only if $\rho^{\varphi}((\y,0)) \geq 0$
    \item $\rho^{\pi}((\y,t)) = -g^{\pi}(\y_t)$\hl{, since $\y \vDash \pi$ if $g^\pi(\y_0) \leq 0$}
    \item $\rho^{\varphi_1 \land \varphi_2}( (\bm{y},t) ) = \min\big(\rho^{\varphi_1}( (\bm{y},t) ), \rho^{\varphi_2}( (\bm{y},t) ) \big)$
    \item $\rho^{\varphi_1 \lor \varphi_2}( (\bm{y},t) ) = \max\big(\rho^{\varphi_1}( (\bm{y},t) ), \rho^{\varphi_2}( (\bm{y},t) ) \big)$
    \item $\rho^{\eventually_{[t_1,t_2]} \varphi}( (\bm{y},t) ) = \max_{t'\in[t+t_1, t+t_2]}\big(\rho^\varphi( (\bm{y},t') )\big)$
    \item $\rho^{\always_{[t_1,t_2]} \varphi}( (\bm{y},t) ) = \min_{t'\in[t+t_1, t+t_2]}\big(\rho^\varphi( (\bm{y},t') )\big)$
    \item $\rho^{\varphi_1 \until_{[t_1,t_2]} \varphi_2}( (\bm{y},t) ) =  \max_{t'\in[t+t_1, t+t_2]}\bigg( \\ \min\Big(\Big[\rho^{\varphi_1}( (\bm{y},t') ), \min_{t'' \in[t+t_1,t']}\big(\rho^{\varphi_2}( (\bm{y},t'') )\big)\Big]\Big)\bigg)$.
\end{itemize}


For a given STL formula $\varphi$, we denote the \hl{time horizon} (minimum number of timesteps after which a signal's satisfaction status is fixed) as $T$.
\hl{
We denote the number of predicates as $N^\pi$ and the number of disjunctive subformulas as $N^\lor$:}

\hl{
\begin{definition} Disjunctive (sub)formulas are defined as follows:
    \begin{itemize}
        \item $\pi$, $\bigwedge_{j=1}^n \varphi_i$, and $\eventually_{[t_1,t_2]} \varphi_i$ are not disjunctive;
        \item $\bigvee_{j=1}^k \varphi_i$ is disjunctive, with $k$ disjunctions;
        \item $\eventually_{[t_1,t_2]} \varphi$ is disjunctive, with $t_2-t_1$ disjunctions;
        \item $\varphi_1 \until_{[t_1,t_2]} \varphi_2$ is disjunctive, with $t_2-t_1$ disjunctions.
    \end{itemize}
\end{definition}}

\hl{ 
We denote the number of disjunctions associated with the $i^{th}$ disjunctive subformula as $N_i$. For example, the formula $\left( a \lor b \right) \land \eventually_{[0,T]} c$ has two disjunctive subformulas, with $N_1 = 2$ and $N_2 = T+1$ respectively.}

\subsection{Problem Statement}\label{subsec:problem}

Given an initial state, an STL specification, and a quadratic running cost, our problem is to find a minimum cost trajectory that satisfies the STL specification. More formally, this problem can be stated as an optimization problem:
\begin{problem}\label{prob:main}
    Given a system of the form (\ref{eq:sys}), an initial state $\x_0$, and an STL specification $\varphi$, solve the optimization problem
    \begin{subequations}\label{eq:main_optimization}
    \begin{align}
        \min_{\x,\u,\y} ~& -\rho^{\varphi}(\y) + \sum_{t=0}^T \x_t^T\Q \x_t + \u_t^T\R\u_t \label{eq:quadratic_cost}\\
        \mathrm{s.t.~} & \x_{t+1} = \A \x_t + \B \u_t \label{eq:state_dynamics}\\
                     & \y_t = \C \x_t + \D \u_t \label{eq:output_dynamics}\\
                     & \x_0 \mathrm{~fixed} \label{eq:initial_state}\\
                     & \x_t \in {X}, \u_t \in {U}, \y_t \in {Y} \label{eq:set_containment}\\
                     & \rho^{\varphi}(\y) \geq 0 \label{eq:positive_robustness}
    \end{align}
    \end{subequations}
    
    where $\Q \succeq 0$ and $\R \succeq 0 $ are symmetric cost matrices. 
\end{problem}

Apart from the robustness measure $\rho^{\varphi}(\y)$, (\ref{eq:main_optimization}) is a convex optimization problem, for which many mature solvers exist \cite{boyd2004convex}. With this in mind, our primary goal is to introduce mixed-integer constraints to define $\rho^{\varphi}(\y)$, preferably using as few binary variables as possible.

\begin{remark}
    Note that with $\Q = \R = 0$, this optimization problem maximizes the robustness measure $\rho$, resulting in a maximally satisfying solution. Similarly, the cost matrices $\Q$ and $\R$ can be scaled to trade off maximizing STL satisfaction and minimizing the running cost. 
\end{remark}

\section{Main Results}\label{sec:main}

The current state-of-the-art is to encode (\ref{eq:main_optimization}) as an MICP with $T N^{\pi}$ binary variables (one per predicate per timestep) \cite{belta2019formal}. In this section, we propose an encoding with $\sum_{i=1}^{N^\lor} \lceil \log_2(N_i+1) \rceil$ binary variables, where $N_i$ is the number of disjunctions associated with the $i^{th}$ disjunctive subformula. \hl{Note that $N_i$ may vary with the time horizon $T$, depending on the specification}.

We first introduce a tree data structure for representing STL formulas. Then, after presenting our proposed encoding, we show how simple offline operations on this data structure can further reduce the number of binary variables. 

\subsection{Representing STL Specifications}\label{subsec:tree_structure}

In this section, we present a tree data structure for STL formulas. A similar approach was taken in \cite{leung2020back}, with the primary aim of enabling efficient gradient-based optimization. Here, our primary aim is to enable more efficient MICP. 

Our tree data structure is formally defined as follows:
\begin{definition}
    An STL Tree $\mathcal{T}^\varphi$ is a tuple $(S, \tau, \circ, \rho^{\varphi})$, where
    \begin{itemize}
        \item $S = [\mathcal{T}^{\varphi_1},\mathcal{T}^{\varphi_2},\dots,\mathcal{T}^{\varphi_N}]$ is a list of $N$ subtrees (i.e., children) associated with each subformula;
        \item $\tau = [t^{\varphi_1}, t^{\varphi_2}, \dots, t^{\varphi_N}]$ is a list of time steps corresponding to the $N$ subformulas;
        \item $\circ \in \{\land, \lor\}$ is a combination type;
        \item $\rho^{\varphi}$ is the STL robustness measure: a function which maps signals $\y$ to scalar values.
    \end{itemize}
\end{definition}

All nodes in an STL Tree are themselves STL Trees, and correspond to subformulas. Leaves of an STL Tree correspond to predicates $\pi$. For a given STL formula $\varphi$, a corresponding tree $\mathcal{T}^{\varphi}$ can be built up recursively by following the syntax presented in Section~\ref{subsec:stl}.

\begin{example}
    Consider the specification $\varphi = \always_{[1,T]}a \land \eventually_{[1,T]} b$, where $a$ and $b$ are predicates. The associated STL Tree is shown in Figure~\ref{fig:simple_tree}. There are two $\land$-type nodes (associated with $\varphi$ and $\always a$), one $\lor$-type node ($\eventually b$), and 2T predicates.
\end{example}

Roughly speaking, the original MICP encoding of \cite{raman2014model} introduces a binary variable for each node in this tree, while \cite{sadraddini2018formal} introduces a binary variable for each leaf\footnote{More precisely, \cite{sadraddini2018formal} introduces a binary variable for each predicate at each timestep. This is typically slightly more than the number of leaves. }. Our proposed encoding, outlined in the following section, instead introduces binary variables only for disjunctive nodes. Furthermore, each disjunctive node only introduces $\lceil \log_2(N+1) \rceil$ binary variables, where $N$ is the number of subformulas.

\subsection{New Mixed-Integer Encoding}\label{subsec:new_encoding}

We now present our proposed mixed-integer encoding. Given an STL formula $\varphi$, we start by constructing a corresponding STL Tree $\mathcal{T}^\varphi$. We define a continuous variable $\rho \geq 0$ which, with some liberty of notation, will provide a lower bound on the robustness measure $\rho^{\varphi}(\y)$. \hl{Note that if the robustness measure is negative, a violation of the constraint $\rho \geq 0$ will render the optimization problem infeasible.}
%

\hl{For each node $\mathcal{T}^\phi$ in the tree, we introduce a \textit{continuous} variable $z^\phi \in [0, 1]$. Ultimately, we will add constraints such that $z^\phi = 1$ if subformula $\phi$ is enforced and $z^\phi=0$ otherwise. While these variables take binary values, declaring them as continuous improves MICP efficiency, since MICP complexity is dominated by the number of binary variables}.

We start with the leaf nodes. Each leaf is associated with a predicate ($\pi$) and a timestep $t$. For each leaf we constrain
\begin{equation}\label{eq:bigM}
    \rho \leq -g^\pi(\y_t) + M(1-z^\pi),
\end{equation}
where $M$ is a large scalar constant. This ``big-M'' constraint enforces $\rho \leq -g^\pi(\y_t)$ if $z^\pi = 1$, but leaves $\rho$ and $\y_t$ essentially unconstrained if $z^\pi = 0$ \cite{conforti2014integer}.

Next we consider all of the $\land$-type nodes in the tree. For the corresponding formula $\phi$ to hold, all of the subformulas $\phi_1,\phi_2,\dots,\phi_N$ must also hold. \hl{In other words, $z^\phi = 1$ implies $z^{\phi_i} = 1$}. This can be encoded without adding binary variables, simply by applying the linear constraints:
\begin{equation}\label{eq:conjunction}
    z^{\phi} \leq z^{\phi_i} \quad \forall i = 1,2,\dots,N.
\end{equation}

Disjunctive constraints present more of a challenge. For these $\lor$-type constraints, we need to enforce at least one $z^{\phi_i} = 1$ if $z^{\phi} = 1$. A natural way to do this would be to define $z^{\phi_i}$ as binary variables and add the constraint: 
\begin{equation*}
    z^{\phi} \leq \sum_{i=1}^N z^{\phi_i}.
\end{equation*}
This approach was taken in \cite{raman2014model}, and adds $N$ binary variables. 

We now show how results in disjunctive programming \cite{vielma2011modeling} can be leveraged to encode disjunctive constraints using only $\log_2(N+1)$ binary variables. To introduce these constraints, we first define SOS1 sets:

\begin{definition}
    A vector $\lambda = [\lambda_1,\lambda_2,\dots,\lambda_n]^T$ is a Special Ordered Set of Type 1 (SOS1)
    if the following conditions hold:
    \begin{itemize}
        \item $\lambda_i \geq 0$
        \item $\sum_{i} \lambda_i = 1$
        \item $\exists ~ j \in [1,...,n]$ s.t. $\lambda_j = 1$
    \end{itemize}
\end{definition}
In other words, a vector $\lambda$ is in SOS1 if it contains exactly one nonzero element, and that element is equal to 1.

Perhaps surprisingly, we can constrain $\lambda \in SOS1$ using a logarithmic number of binary variables and constraints:
\begin{theorem}[\cite{vielma2011modeling}]\label{theorem:sos1}
    Assume that $n$ is a power of 2. Let $I = {1,2,\dots,n}$ and $B : I \to \{0,1\}^{\log_2(n)}$ be any bijective function. Then the following constraints enforce $\lambda \in SOS1$:
    \begin{gather*}
        \lambda_i \geq 0, \quad
        \sum_{i} \lambda_i = 1, \\
        \sum_{j \in J^+(k,B)} \lambda_j \leq \zeta_k, \quad
        \sum_{j \in J^0(k,B)} \lambda_j \leq (1-\zeta_k)\\
        \zeta_{k} \in \{0,1\} \quad \forall k \in [1,2,...,\log_2(n)] 
    \end{gather*}
    where $J^+(k,B) = \{i \mid k \in \text{supp}(B(i)) \}, J^0(k,B) = \{i \mid k \notin \text{supp}(B(i))\}$\hl{, and $\text{supp}(B(i))$ denotes the support of $B(i)$.}
\end{theorem}

For the bijective function $B(i)$, the mapping from index $i$ to its binary representation is a natural choice. In that case, the values of binary variables $\zeta_k$ essentially represent which index $i$ is holds the nonzero value. If $n$ is not a power of 2, we can simply add $e^{\lceil \log_2 n\rceil} - n$ elements along with linear constraints $\lambda_i = 0$ that force these elements to be zero.

This brings us back to encoding disjunctive STL subformulas. For these $\lor$-type nodes in the STL tree, we need $z^\phi \implies \bigvee_{i=1}^N z^{\phi_i}$. This can be written as a SOS1 constraint
\begin{equation}\label{eq:disjunction}
    [1-z^\phi,z^{\phi_1},z^{\phi_2},\dots,z^{\phi_N}] \in SOS1,
\end{equation}
which can be encoded with $\lceil \log_2(N+1) \rceil$ binary variables.

\hl{ Note that SOS1 ensures \textit{exactly} one nonzero element, while disjunction requires \textit{at least} one nonzero element. This is not an issue, however, because we only need a sufficient (not necessary) condition for $\y \vDash \phi$.  Focusing on specifications in PNF allows us to guarantee both soundness and completeness while only requiring this sort of sufficient condition \cite{belta2019formal,sadraddini2018formal}. 
}

Finally, for the root node we constrain
\begin{equation}\label{eq:z_eq_1}
    z^{\varphi} = 1,
\end{equation}
to ensure that the overall specification $\varphi$ must be satisfied. 

\begin{proposition}
    \hl{If the} constraints (\ref{eq:bigM})-(\ref{eq:z_eq_1}) are met, then
    the following properties hold:
    \begin{enumerate}
        \item The \hl{continuous} variable $z^{\phi}$ is 1 only if $\y \vDash \phi$;
        \item The decision variable $\rho$ provides a lower bound for the robustness measure $\rho^{\varphi}(\y)$;
        \item The constraints (\ref{eq:bigM})-(\ref{eq:z_eq_1}) are feasible if and only if $\y \vDash \varphi$.
    \end{enumerate}
\end{proposition}

\begin{proof}
    \hl{Property~1 follows naturally from Theorem \ref{theorem:sos1} and the reasoning of \cite{raman2014model}. Note that this property is one-directional: it may be the case that $z^\phi = 0$ and $\y \vDash \phi$ for some subformula.}
    
    \hl{
    Property~2 follows from the definition of the robustness measure and constraints (\ref{eq:bigM}). }
    
    \hl{
    For Property~3, the reverse direction (only if) follows directly from Property~1 and constraint (\ref{eq:z_eq_1}). The forward direction (if) follows from Property~2 and the constraint $\rho \geq 0$.}
\end{proof}

This allows us to solve Problem~\ref{prob:main} as an MICP:

\begin{subequations}\label{eq:our_micp}
\begin{align}
    \min_{\x,\u,\y,z^{\phi_i},\rho} ~& -\rho + \sum_{t=0}^T \x_t^T\Q \x_t + \u_t^T\R\u_t \\
    \text{s.t. } & \text{Dynamics constraints (\ref{eq:state_dynamics}-\ref{eq:set_containment})} \\
                 & \text{STL Constraints (\ref{eq:bigM}-\ref{eq:z_eq_1})} \\
                 & \rho \geq 0
\end{align}
\end{subequations}

Similar to the existing MICP encodings \cite{raman2014model,sadraddini2015robust,sadraddini2018formal,belta2019formal}, this MICP is sound, complete, and globally optimal:

\begin{theorem}[Sound and Complete]
    A dynamically feasible output signal $\y$ is a solution to Problem~\ref{prob:main} if and only if $\y$ is a solution to (\ref{eq:our_micp}).
\end{theorem}

\begin{theorem}[Globally Optimal]
    If Problem~\ref{prob:main} is feasible, the MICP (\ref{eq:our_micp}) finds a globally optimal solution. 
\end{theorem}

The advantage of our approach over other sound, complete, and globally optimal methods \cite{raman2014model,sadraddini2015robust,sadraddini2018formal,belta2019formal} is that our approach typically uses fewer binary variables. For many specifications, like that shown in Figure~\ref{fig:simple_tree} and the examples in Section~\ref{sec:experimenal}, this difference is quite significant. 

It is possible, however, for our approach to have more binary variables than the standard encoding, particularly for specifications with few predicates but many nested disjunctions. To this end, the next section shows how the STL Tree data structure can be used to further reduce the MICP size. 

\subsection{Formula Flattening}

Unlike the state-of-the-art MICP encoding with binary variables for every predicate at every timestep \cite{belta2019formal}, the number of binary variables in our proposed encoding depends heavily on the structure of the STL Tree $\mathcal{T}^\varphi$. Since logically equivalent formulas can have different tree structures, this can have a major impact on solver efficiency.

For example, consider the logically equivalent formulas $\varphi_1 = a \lor b \lor c$ and $\varphi_2 = a \lor (b \lor c)$, where $a$, $b$, and $c$ are predicates. STL Trees are illustrated in Figure~\ref{fig:equivalent_trees}. $\varphi_1$ includes one disjunctive subformula, and can be encoded with 2 binary variables. The logically equivalent $\varphi_2$ includes two disjunctive subformulas, however, and requires 4 binary variables. 

\begin{figure}
    \centering
    \begin{subfigure}[b]{0.48\linewidth}
        \centering
        \begin{tikzpicture}
            \tikzset{sibling distance=20pt}
            \Tree [.$\varphi_1$ $a$ $b$ $c$ ]
        \end{tikzpicture} 
        \caption{$\varphi_1 = a \lor b \lor c $}
        \label{fig:phi_1}
    \end{subfigure}
    \begin{subfigure}[b]{0.48\linewidth}
        \centering
        \begin{tikzpicture}
            \tikzset{level distance=20pt, sibling distance=20pt}
            \Tree [.$\varphi_2$
                    [.$a$ ]
                    [.{$b\lor c$} $b$ $c$ ]
                  ]
        \end{tikzpicture} 
        \caption{$\varphi_2 = a \lor (b \lor c)$}
        \label{fig:phi_2}
    \end{subfigure}
    \caption{STL Trees for two logically equivalent formulas. Our approach introduces 2 binary variables for $\varphi_1$, but 4 variables for $\varphi_2$.}
    \label{fig:equivalent_trees}
\end{figure}
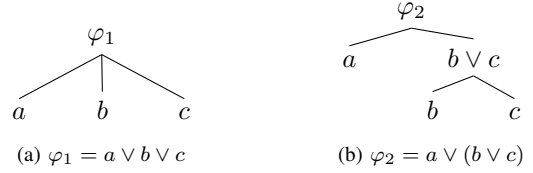

\begin{figure*}
    \centering
    \begin{subfigure}{0.22\linewidth}
        \centering
        \includegraphics[width=1.0\linewidth]{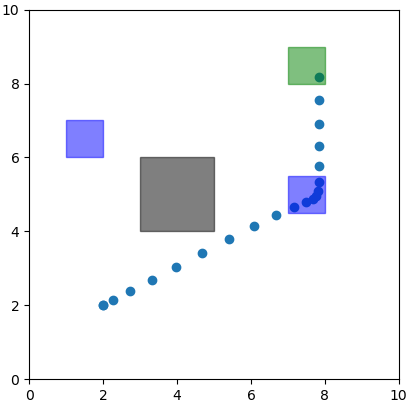}
        \caption{\hl{Two-Target Reach-Avoid (\ref{eq:either_or})}}
        \label{fig:either_or}
    \end{subfigure}
    \begin{subfigure}{0.22\linewidth}
        \centering
        \includegraphics[width=1.0\linewidth]{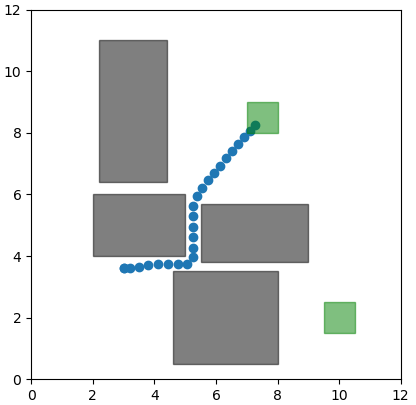}
        \caption{Narrow Passage (\ref{eq:narrow_passage})}
        \label{fig:narrow_passage}
    \end{subfigure}
    \begin{subfigure}{0.22\linewidth}
        \centering
        \includegraphics[width=1.0\linewidth]{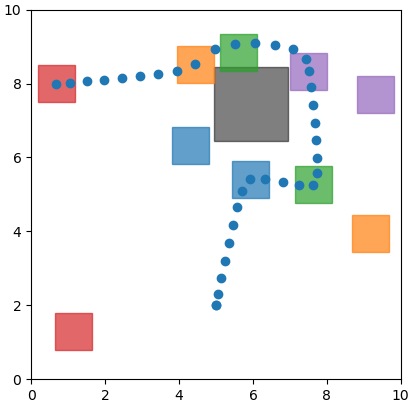}
        \caption{Many-Target (\ref{eq:multitarget})}
        \label{fig:multitarget}
    \end{subfigure}
    \begin{subfigure}{0.32\linewidth}
        \centering
        \includegraphics[width=1.0\linewidth]{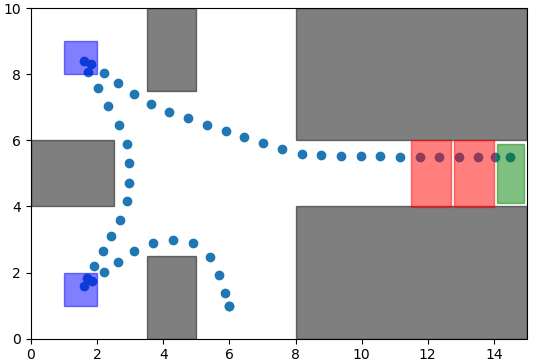}
        \caption{Door Puzzle (\ref{eq:door_puzzle})}
        \label{fig:door_puzzle}
    \end{subfigure}
    \caption{Illustrations of example scenarios along with solutions generated by our proposed approach. Solve times are shown in Table~\ref{tab:solve_times}. }
    \label{fig:scenarios}
    \vspace{-1em}
\end{figure*}

Clearly, we should prefer ``flatter'' STL Trees for our proposed encoding. To this end, \hl{specification flattening} provides an automatic means of compressing formulas with many layers into logically equivalent formulas with fewer layers. The basic idea is to search for nodes which have the same combination type ($\land$ or $\lor$) as their parent. If this is the case, that node's children can be moved up a level and the node removed. 


While formula flattening has the most dramatic impact on our proposed approach, it can also benefit other MICP encodings \cite{raman2014model,belta2019formal}, since flatter formulas introduce fewer continuous variables and constraints.
Interestingly, while formula flattening consistently reduces the number of binary variables, it does not always improve solve times. We attribute this counter-intuitive result to the fact that certain solver-specific heuristics may be more effective on unflattened formulas. In Section~\ref{sec:experimenal}, we only flatten formulas if it improves performance.

\section{Simulation Experiments}\label{sec:experimenal}

In this section, we demonstrate the improved scalability of our proposed approach. This is largely due to a reduction in the number of binary variables --- especially for specifications involving $\eventually$ and $\until$ over long time-horizons. 

We consider several benchmark specifications involving a simple mobile robot exploring a planar environment. These specifications are illustrated in Figure \ref{fig:scenarios}. In all scenarios, we assume that the robot is governed by double-integrator dynamics in both the horizontal and vertical directions, i.e.,
\begin{equation*}
    \x = [p_x, p_y, \dot{p}_x, \dot{p}_y]^T \quad \u = [\ddot{p}_x,\ddot{p}_y]^T
\end{equation*}
where $p_x$ is the horizontal position of the robot and $p_y$ is the vertical position, in meters. The output signal is defined as $\y = [p_x, p_y]^T$, giving us dynamics
\begin{equation*}
    \A = \begin{bmatrix}\bm{I} & \bm{I} \\ \bm{0}  & \bm{I} \end{bmatrix}, \quad \B = \begin{bmatrix} \bm{0} \\ \bm{I}\end{bmatrix}, \quad \C = \begin{bmatrix} \bm{I} & \bm{0}\end{bmatrix}, \quad \bm{D} = \bm{0}.
\end{equation*}
State and input constraints ${X}$ and ${U}$ are defined to enforce maximum acceleration and velocity of 0.5 $m/s^2$ and 1 $m/s$:
\begin{gather}
    {X} = \{ \x \mid 0 \leq p_x \leq 15, 0 \leq p_y \leq 15, |\dot{p}_x| \leq 1, |\dot{p}_y| \leq 1 \}. \\
    {U} = \{ \u \mid |\ddot{p}_x| \leq 0.5, |\ddot{p}_y| \leq 0.5 \}.
\end{gather}

We used python and the Drake \cite{drake} mathematical programming interface to set up the MICP (\ref{eq:our_micp}), and solved the MICP with Gurobi \hl{version 9.5.1} \cite{gurobi} with default options. As a baseline, we compare to the state-of-the-art MICP described in \cite{belta2019formal}, similarly implemented in Drake/Gurobi. All experiments were performed on a laptop with an Intel i7 processor and 32GB RAM. \hl{Open-source software is available online \cite{stlpy}}.

In the first scenario, shown in Figure~\ref{fig:either_or}, our mobile robot must avoid an obstacle ($O$, grey), \hl{visit one of two intermediate targets ($T_1$, $T_2$, blue) for at least 5 timesteps}, and reach a goal ($G$, green). This specification can be written as 
\begin{equation}\label{eq:either_or}
    \hl{\eventually_{[0,T-5]} ( \always_{[0,5]} T_1 \lor  \always_{[0,5]} T_2) \land \always_{[0,T]} \lnot O \land \eventually_{[0,T]}G,}
\end{equation}
where $T$ is the specification time-bound. 
 
The second scenario involves many obstacles and narrow passages, and is shown in Figure~\ref{fig:narrow_passage}. The robot must eventually reach one of two possible goals in addition to avoiding obstacles. This specification can be written as:
\begin{equation}\label{eq:narrow_passage}
    \eventually_{[0,T]} (G_1 \lor G_2) \land \always_{[0,T]} \left( \bigwedge_{i=1}^4 \lnot O_i \right)
\end{equation}
 
The third scenario, shown in Figure~\ref{fig:multitarget}, is designed to contain a more complex logical structure. In addition to avoiding an obstacle, the robot must visit several groups of targets $T^j$ (red, green, blue, orange, and purple). There are two targets in each group, and at least one target from each group must be visited. This specification can be written as
\begin{equation}\label{eq:multitarget}
    \bigwedge_{i=1}^{5} \left( \bigvee_{j=1}^{2} \eventually_{[0,T]} T_{i}^{j} \right) \land \always_{[0,T]} (\lnot O),
\end{equation}
where $T_i^j$ denotes the $j^{th}$ target in group $i$.

The final scenario is inspired by \cite{sun2022multi,vega2018admissible}, and requires the robot to collect keys (i.e., visit blue regions) associated with certain doors (red) before reaching an end goal (green). This specification can be written as
\begin{equation}\label{eq:door_puzzle}
    \bigwedge_{i=1}^{2} \left( \lnot D_i \until_{[0,T]} K_i \right)  \land \eventually_{[0,T]} G \land
                \always_{[0,T]} (\bigwedge_{i=1}^5 \lnot O_i),
\end{equation}
where $K_i$ denotes picking up the $i^{th}$ key and $D_i$ denotes passing through the $i^{th}$ door region.

Solve times and number of binary variables for each of these specifications over several time horizons $T$ are shown in Table~\ref{tab:solve_times}. For all of these examples, our proposed approach uses fewer binary variables than the standard MICP \cite{belta2019formal}. The standard MICP is often faster for shorter time horizons, but our proposed encoding is consistently faster for long time horizons. 

\begin{figure}
    \centering
    \includegraphics[width=0.6\linewidth]{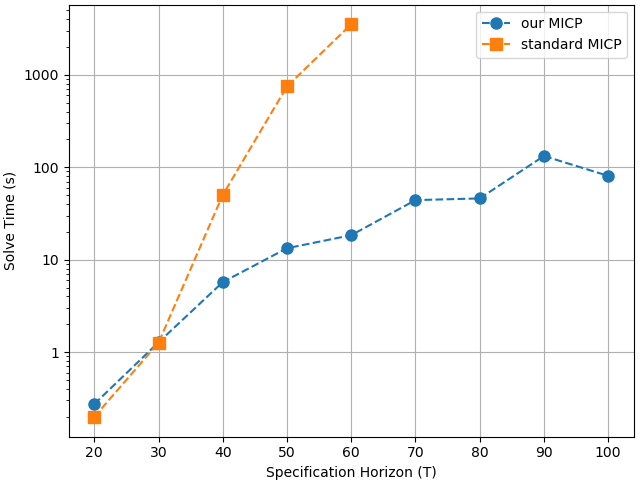}
    \caption{Solve times for specification (\ref{eq:either_or}) over different time horizons.}
    \label{fig:time_scalability}
\end{figure}

For specification (\ref{eq:either_or}), we provide a graph of solve times versus time horizon in Figure~\ref{fig:time_scalability}. Again, we see that our proposed approach outperforms the state-of-the-art for long time horizons, with the performance gap between the two approaches growing as $T$ increases. This makes sense given how each approach introduces binary variables to handle the $\eventually$ subformulas. Under the standard approach, the number of binary variables used to encode $\eventually_{[0,T]}$ increases linearly with $T$. Under our proposed approach, the number of binary variables increases only logarithmically. 

\begin{table}[]
    \centering
    \hl{
    \begin{tabular}{c|c||c|c|c|c}
        Specification & Horizon (T) & 
        \multicolumn{2}{c|}{ Binary Vars. } & 
        \multicolumn{2}{c}{ Solve Time (s) } \\
        \hline
        & & \cite{belta2019formal} & Ours & \cite{belta2019formal} & Ours \\
         \hline
         \hline
         \hl{Two-Target} (\ref{eq:either_or}) & 25 & 1216 &  89 & 0.662 & \textbf{0.579} \\
                                              & 50 & 2616 & 166 &   749 & \textbf{13.3} \\
         \hline
         \hl{Narrow Passage}       & 25 &  624 & 318 & \textbf{2.27} & 5.92 \\
         (\ref{eq:narrow_passage}) & 50 & 1124 & 619 &         168.8 & \textbf{158.1} \\
         \hline
         \hl{Many-Target} (\ref{eq:multitarget}) & 25 & 1144 & 441 & \textbf{3.36} & 16.8  \\
                                                 & 50 & 2244 & 846 &       $>$7500 & \textbf{2666} \\
         \hline
         \hl{Door Puzzle} (\ref{eq:door_puzzle}) & 25 &  3432 & 2355 & \textbf{8.46} & 8.75 \\
                                                 & 50 & 11832 & 8433 &       $>$7500 & \textbf{1392} \\
    \end{tabular}}
    
    \caption{Solve Times for Benchmark Scenarios}
    \label{tab:solve_times}
\end{table}

Finally, we note that the standard MICP's superior performance over short time horizons may be due to Gurobi's presolve routines taking advantage of the extra binary variables to perform additional simplifications. Take, for example, the narrow passage specification (\ref{eq:narrow_passage}) with \hl{25 timesteps. With presolve enabled, the standard approach outperforms our approach (Table~\ref{tab:solve_times}). With presolve disabled, however, the standard approach takes over an hour to find an optimal solution, while our approach takes about a minute. }

\section{Conclusion}\label{sec:conclusion}

We proposed a more efficient mixed-integer encoding for Signal Temporal Logic. Our primary insight is that we can encode disjunctions using a logarithmic number of binary variables, and conjunctions without introducing any binary variables. We also introduced a tree data structure for STL, and showed how this data structure can be exploited to further reduce the number of binary variables and constraints. We demonstrated the improved scalability of our approach in simulation examples, where our proposed encoding outperformed the state-of-the-art by orders of magnitude on long and complex specifications.

\bibliographystyle{IEEEtran}
\bibliography{references}

\end{document}